\documentclass[11pt,reqno,a4paper]{amsart}

\usepackage{array}
\usepackage{enumitem}
\usepackage{mathtools}
\usepackage{pgfplots}
\usepackage{soul}

\usepackage{adjustbox} 

\usepackage{amsmath,amssymb,amsfonts}
\usepackage{tikz}
\usepackage{hyperref}
\usepackage{appendix}
\usepackage{scalerel}
\usepackage{graphicx}
\usepackage{subcaption}
\usepackage{amsthm}
\usepackage{multirow}
\usepackage{pdflscape}
\usepackage{afterpage}
\usepackage{capt-of} 
\usepackage{lipsum}
\usepackage{booktabs}
\usepackage{siunitx}
\usepackage{esvect}

\tikzset{decorated arrows/.style={
    postaction={
        decorate,
        decoration={
            markings,
            mark=between positions 0 and 1 step 15mm with {\arrow[black]{stealth};}
            }
        },
    }
}
 
\tikzset{decorated arrows2/.style={
    postaction={
        decorate,
        decoration={
            markings,
            mark=at position 15mm with {\arrow[black]{stealth};}
            }
        },
    }
}

\usepgfplotslibrary{colormaps}
\usetikzlibrary{arrows}

\makeatletter
\def\set@curr@file#1{%
  \begingroup
    \escapechar\m@ne
    \xdef\@curr@file{\expandafter\string\csname #1\endcsname}%
  \endgroup
}
\def\quote@name#1{"\quote@@name#1\@gobble""}
\def\quote@@name#1"{#1\quote@@name}
\def\unquote@name#1{\quote@@name#1\@gobble"}
\makeatother
\usepackage{graphics}

\newtheorem{theorem}{Theorem}
\newtheorem{lemma}{Lemma}

\newcommand{\dpt}{\displaystyle}

\newcommand{\RN}[1]{%
  \textup{\uppercase\expandafter{\romannumeral#1}}%
}
      
\author[J. P. S. M. de Carvalho]{Jo\~ao P. S. Maur\'icio de Carvalho \\
\MakeLowercase{jp.carvalho@upt.pt}  \\
\\
Centre for Mathematics, University of Porto, \\ Rua do Campo Alegre s/n, Porto 4169-007, Portugal \\
\\
Prince Henry Portucalense University, \\ Rua Dr. Ant\'onio Bernardino de Almeida 541, Porto 4200-072, Portugal \\
}

\begin{document}


\title[Qualitative analysis of HAART effects on HIV and SARS-CoV-2 coinfection]
{Qualitative analysis of HAART effects \\ on HIV and SARS-CoV-2 coinfection}

\date{\today}

\begin{abstract}
HIV is known for causing the destruction of the immune system by affecting different types of cells, while SARS-CoV-2 is an extremely contagious virus that leads to the development of COVID-19. In this study, we propose a mathematical model to investigate the interaction between HIV and SARS-CoV-2 under highly active antiretroviral therapy (HAART). We determine the conditions for the endemic equilibria of both viruses, showing that transcritical bifurcations occur when the basic reproduction numbers of HIV and SARS-CoV-2 pass through 1. We set the condition for the stability of the disease-free equilibrium point of the model with coinfection as a function of the basic reproduction number $\mathcal{R}_0$. Through numerical simulations, we conclude that HAART, used to control HIV, also reduces the proliferation of SARS-CoV-2-infected cells in coinfected hosts. These findings provide important insights into the epidemiological dynamics of HIV and SARS-CoV-2 coinfection.
\end{abstract}

\maketitle


\section{Introduction}

\label{intro}
\noindent The first official case of CoViD-19 in China emerged in December 2019 and was linked to the Huanan seafood market, in Wuhan \cite{Hu2021}. Since then SARS-CoV-2 has reached 219 countries and territories having infected more than 750 million people, thus becoming a global pandemic and causing more than 7 million deaths worldwide \cite{WHO2021}. The most frequently noticed symptoms in infected people are fever, cough and respiratory disorders \cite{Vitiello2020,Sun2020}. The most affected people are elderly and adults over 60 years old and/or those with comorbidities, such as obesity, diabetes, oncological diseases, heart problems, among others \cite{RothanByrareddy2020,Moriconi2020,Bajgain2021}.

\smallskip

\noindent HIV continues to be a global health challenge, with 42.3 million deaths to date and ongoing transmission around the world. According to the World Health Organization (WHO), at the end of 2023, around 39.9 million people were living with HIV, of whom approximately 75\% have access to antiretroviral treatment \cite{Basoulis2023} and 65\% of them in the WHO African Region. That same year, 1.3 million new infections were recorded and 630\,000 people died from HIV-related causes. Although there is no cure for HIV, access to prevention, diagnosis and treatment helps people manage the virus as a long-term disease. World health organizations are working to end the HIV epidemic by 2030 \cite{WHO_HIV}.

\smallskip

\noindent Coinfection by HIV and SARS-CoV-2 raises a significant public health concern, since both diseases affect the immune system in distinct and complex ways. HIV progressively weakens the body's immune defenses, while SARS-CoV-2 can trigger severe and potentially dysregulated immune responses. In seropositive individuals, the interaction between these two viruses can aggravate the clinical progression of COVID-19 \cite{Hoft2024}, increase susceptibility to opportunistic infections and affect the effectiveness of antiviral treatments.

\smallskip

\noindent From a mathematical perspective, studying the dynamics of this coinfection is essential to understanding how the two viruses interact and impact the health outcomes of patients. Mathematical models allow us to explore the progression of infection in specific populations, predict outbreak scenarios and assess the impact of different interventions, such as combination therapies and vaccination campaigns \cite{Brauer2012}. For individuals living with HIV and contracting COVID-19, these models are particularly valuable since they help identify factors that can exacerbate the clinical picture and assist in the design of more effective treatment strategies, both from an immunological and epidemiological point of view \cite{Hoft2024}.



\medskip

\noindent {\bf State of the art.}
  \noindent In 2022, Mekonena and Obsu \cite{MekonenaObsu2022} analyzed a model for TB-COVID-19 coinfection, compartmentalizing the population into seven classes. The authors computed the basic reproduction numbers for each disease, showing that the disease-free and endemic equilibria remain stable (or unstable) if these numbers are lower (or higher) than one. The sensitivity analysis suggested that reducing contact rates and increasing the speed of transitions from latent to infected can reduce the spread of both diseases.
  
  \smallskip
  
  \noindent In 2022, Batu {\it et al.}~\cite{Batu2023} developed a mathematical model to study the impact of intervention strategies and identify mortality risk factors in seropositive people infected with COVID-19.  Numerical simulations, based on data from Ethiopia, revealed that vaccination against COVID-19 and increased treatment rates reduce cases of co-infection and the risk of mortality for HIV-infected individuals, underlining the importance of vaccination programs and medical interventions.
  
  \smallskip
  
\noindent   In 2024, Vemparala {\it et al.}~\cite{Vemparala2024} studied a model based on HIV-1 control and remission through mathematical modeling. The authors explored the mechanisms behind natural and post-treatment control of HIV-1, assessed the potential causes of loss of control and quantified the effects of intervention. Their work highlights both the progress achieved in the optimization of the intervention and the ongoing challenges in applying these results in practical contexts.
  
  \smallskip
  
\noindent  During the same year, Chen {\it et al.}~\cite{Chen2024} studied an HIV virus-cell model with intracellular delays, focusing on the interactions between wild-type and drug-resistant HIV strains. The authors set stability criteria based on the basic reproduction number and analyzed Hopf bifurcations, finding that interactions between two strains lead to more complex dynamics, including higher viral loads and potential instability. They found that drug resistance influences the wild-type strain's survival, with an impact on HIV transmission.

\medskip

\noindent {\bf Article structure.}
We describe the model and the population dynamics in Section \ref{SEC_2}. In Section \ref{SEC_3} we prove the positivity and boundedness of the solutions of the coinfection model, analyze the equilibria of the HIV and SARS-CoV-2 submodels and their basic reproduction numbers. We also present our main results, followed by their proofs in Sections \ref{SEC_5} and \ref{SEC_6}. In Section \ref{SEC_4} we carry out a sensitivity analysis of the parameters that constitute the basic reproduction numbers for each submodel. We perform several numerical simulations in Section \ref{SEC_7}, and in Section \ref{SEC_CONCL} we present our conclusions and discuss potential future work.

\medskip

\noindent With this work we aim to understand how HAART therapy affects the dynamics of viral load and cells infected with SARS-CoV-2. For this purpose, we will first perform an analytical study of the model before exploring the interaction of both viruses in the presence of HAART through numerical simulations.


\section{Model formulation}
\label{SEC_2}
\noindent The model we propose is subdivided into four cell populations (cell population $P_C$):

\medskip

\begin{itemize}
\item[$T$:] number of healthy/target cells susceptible to infection;

\bigskip

\item[$I_H$:] number of cells that are currently infected by HIV; 

\bigskip

\item[$I_S$:] number of cells that are currently infected by SARS-CoV-2; 

\bigskip

\item[$C$:] number of cells that are currently coinfected by HIV and SARS-CoV-2,
\end{itemize}

\medskip

\noindent and two classes of virus (virus population $P_V$): 

\medskip

\begin{itemize}
\item[$V_H$:] HIV viral load;

\bigskip

\item[$V_S$:] SARS-CoV-2 viral load.
\end{itemize}

\medskip

\noindent Let

\begin{equation}\label{modelo}
\begin{array}{lcl}
\dot{X} = \mathcal{F}(X) \quad \Leftrightarrow \quad
\begin{cases}
&\dot{T} = \lambda - k_1 \left(1-\epsilon_{RT}\right)TV_H - k_2 T V_S - \mu T \\
\\
&\dot{I}_H = k_1 \left(1-\epsilon_{RT}\right)TV_H - k_2 I_H V_S - \mu I_H \\
\\
&\dot{I}_S = k_2 T V_S - k_1\left(1-\epsilon_{RT}\right) I_S V_H - \mu I_S \\
\\
&\dot{V}_H = n_H \left(1-\epsilon_P\right) \mu I_H - \sigma_H V_H \\
\\
&\dot{V}_S = n_S \mu I_S - \sigma_S V_S \\
\\
&\dot{C} = k_1 \left(1-\epsilon_{RT}\right) I_S V_H + k_2 I_H V_S - \mu C ,
\end{cases}
\end{array}
\end{equation}

\medskip

\noindent be the nonlinear system of ODE that modulates the intereaction between cells and virus, where \smallskip

$$
\begin{array}{lcl}
X(t) &=& \left( T(t), I_H(t), I_S(t), V_H(t), V_S(t), C(t) \right), \\
\\
\dot{X} &=& \left(\dot{T}, \dot{I}_H, \dot{I}_S, \dot{V}_H, \dot{V}_S, \dot{C} \right) \,\,\, = \,\,\, \dpt \left(\frac{\mathrm{d}T}{\mathrm{d}t},\frac{\mathrm{d}I_H}{\mathrm{d}t},\frac{\mathrm{d}I_S}{\mathrm{d}t},\frac{\mathrm{d}V_H}{\mathrm{d}t}, \frac{\mathrm{d}V_S}{\mathrm{d}t}, \frac{\mathrm{d}C}{\mathrm{d}t}\right),
\end{array}
$$

\medskip

\noindent and 

$$\Psi = \left\{ (\lambda,k_1,k_2,\mu,\epsilon_{RT},\epsilon_P,n_H,n_S,\sigma_H,\sigma_S) \in (\mathbb{R}^+)^{10} \right\}$$

\medskip

\noindent is the set of parameters of \eqref{modelo}. The vector field associated to \eqref{modelo} will be called by $\mathcal{F}$ and the associated flow is 

$$\psi \left(t, (T_0, {I_H}_0, {I_S}_0, {V_H}_0, {V_S}_0, C_0) \right)  , \,\, t \in \mathbb{R}_0^+ \, , \quad \left(T_0, {I_H}_0, {I_S}_0, {V_H}_0, {V_S}_0, C_0 \right) \in (\mathbb{R}_0^+)^6 \, .
$$ 

\medskip

\noindent {\bf Dynamics and interpretation of the constants.} The constant production rate of healthy $T$ cells is given by $\lambda$. Healthy T cells are infected by HIV and SARS-CoV-2 at a rate $k_1$ and $k_2$, respectively. The parameter $0 \leq \epsilon_{RT} \leq 1$ represents the efficacy of reverse transcriptase inhibitors (RTI), reaching 100\% effectiveness when $\epsilon_{RT} = 1$. Similarly, $0 \leq \epsilon_{P} \leq 1$ denotes the efficacy of protease inhibitors (PI), with $\epsilon_{P} = 1$ indicating full effectiveness. HIV and SARS-CoV-2 are produced by infected cells with the bursting size $n_H$ and $n_S$, respectively. The natural death rates of $T$, $I_H$, $I_S$ and $C$ cells are given by $\mu_T$, $\mu_H$, $\mu_S$ and $\mu_C$, respectively. However, for convenience in algebraic calculations, we will assume an equal natural death rate $\mu$ for all cells. The death rate of HIV is given by $\sigma_H$. The viral load of SARS-CoV-2 dies at a rate $\sigma_S$. The description and value of these parameters can be found in Table \ref{tabela} and the dynamics of the cell and virus populations are given by

\begin{equation}\label{P_C}
\begin{array}{lcl}
P_C(t) = T(t) + I_H(t) + I_S(t) + C(t)
\end{array}
\end{equation}

\noindent and

\begin{equation}\label{P_V}
\begin{array}{lcl}
P_V(t) = V_H(t) + V_S(t)
\end{array}
\end{equation}

\medskip

\noindent respectively. Figure \ref{diagram} illustrates a diagram depicting the interactions among $P_C(t)$ under the HIV and SARS-CoV-2 viral loads. Moreover, the dynamics of all classes of \eqref{modelo} are given by

\begin{equation*}\label{A_total_V_P}
\begin{array}{lcl}
N(t) = P_C(t) + P_V(t) = T(t) + I_H(t) + I_S(t) + V_H(t) + V_S(t) + C(t) .
\end{array}
\end{equation*} 

\begin{figure}[ht!]
\includegraphics[width=1.19\textwidth]{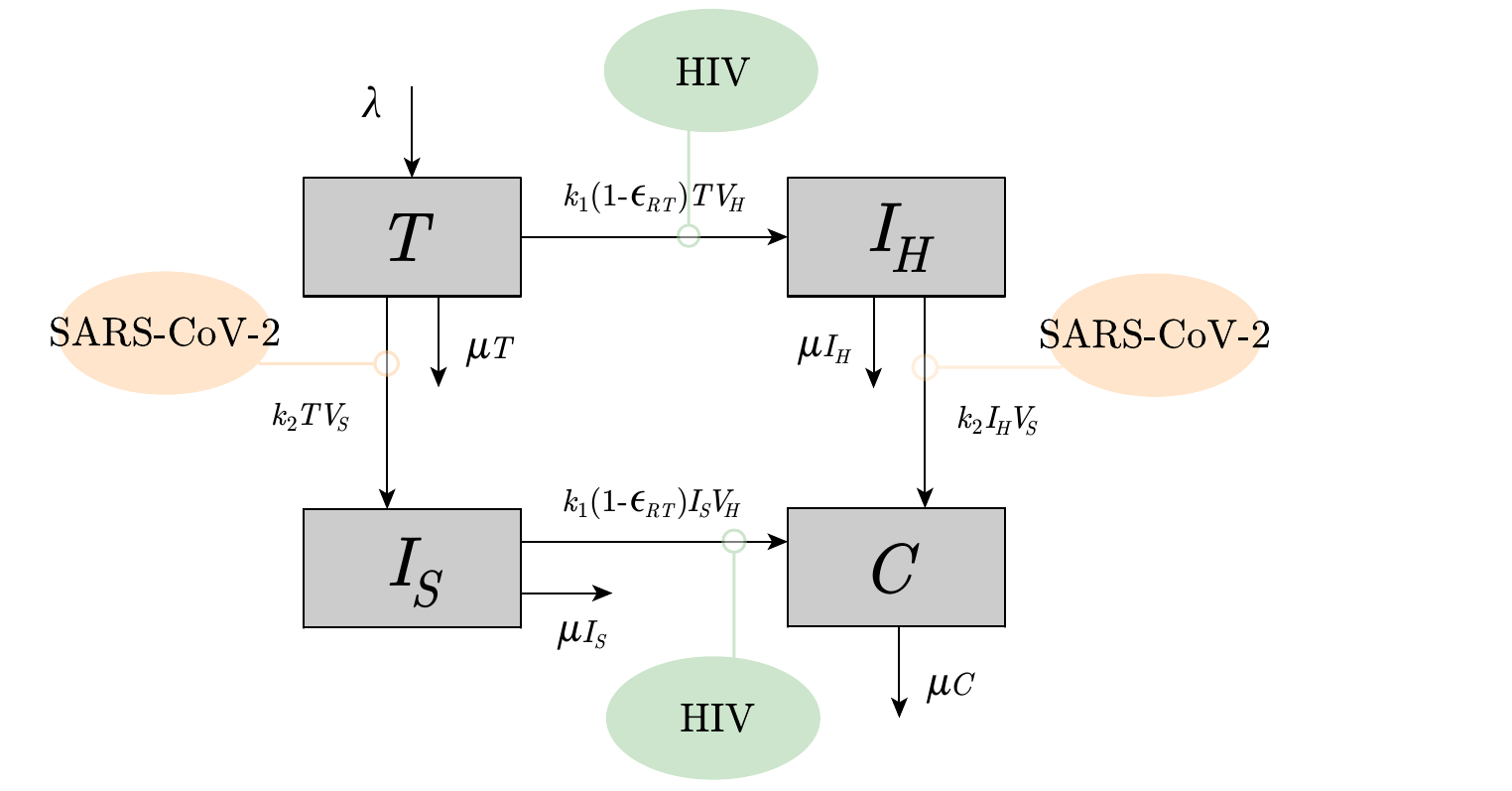}
\caption{Diagram of the interactions between $T$, $I_H$, $I_S$ and $C$ cells of system \eqref{modelo} under the influence of the viral loads HIV and SARS-CoV-2.}
\label{diagram}
\end{figure}


\section{Properties of the model and main results} \label{SEC_3}
\noindent In this section we prove that all solutions of system \eqref{modelo} are positive and bounded. Moreover, we analyze two submodels, HIV and SARS-CoV-2 submodels, derived from system \eqref{modelo}. 

\medskip

\noindent {\bf Positivity and boundedness of solutions.}


\begin{theorem} \label{thm_bound}
The solutions for $T$, $I_H$, $I_S$, $V_H$, $V_S$ and $C$ of model \eqref{modelo} are nonnegative and remain bounded for all $t \geq 0$ within the biologically feasible region defined by the set

\begin{scriptsize}
\begin{equation*}\label{Omega_}
\begin{array}{lcl}
\Omega = \left\{  \left(T, I_H, I_S, V_H, V_S, C \right) \in \left(\mathbb{R}_0^+\right)^6: 0 \, \leq \, P_C \, \leq \, \dfrac{\lambda}{\mu} \,\, , \quad 0 \, \leq \, P_V \, \leq \, \dfrac{n_H \left(1-\epsilon_P\right) \mu I_H^{\text{max}} + n_S \mu I_S^{\text{max}}}{\sigma} \right\}.
\end{array}
\end{equation*} 
\end{scriptsize}

\end{theorem}

\begin{proof}

For the active population and virus of system \eqref{modelo} we have

\smallskip

\begin{equation*}\label{positivity}
\begin{array}{lcl}
\bigskip
\bigskip
\dfrac{\mathrm{d}T}{\mathrm{d}t} \Big|_{T = 0} &=& \lambda \,\,\, > \,\,\, 0 \\
\bigskip
\bigskip
\dfrac{\mathrm{d}I_H}{\mathrm{d}t} \Big|_{I_H = 0} &=& k_1 \left(1-\epsilon_{RT}\right)TV_H \,\,\, \geq \,\,\, 0 \\
\bigskip
\bigskip
\dfrac{\mathrm{d}I_S}{\mathrm{d}t} \Big|_{I_S = 0} &=& k_2 T V_S \,\,\, \geq \,\,\, 0 \\
\bigskip
\bigskip
\dfrac{\mathrm{d}V_H}{\mathrm{d}t} \Big|_{V_H = 0} &=& n_H \left(1-\epsilon_P\right) \mu I_H  \,\,\, \geq \,\,\, 0 \\
\bigskip
\bigskip
\dfrac{\mathrm{d}V_S}{\mathrm{d}t} \Big|_{V_S = 0} &=& n_S \mu I_S  \,\,\, \geq \,\,\, 0 \\
\bigskip
\bigskip
\dfrac{\mathrm{d}C}{\mathrm{d}t} \Big|_{C = 0} &=& k_1 \left(1-\epsilon_{RT}\right) I_S V_H + k_2 I_H V_S \,\,\, \geq \,\,\, 0 .
\end{array}
\end{equation*}


This demonstrates that $\left(\mathbb{R}_0^+\right)^6$ is positively invariant. Now, if we prove that $P_C$ and $P_V$ are both bounded, then $N$ is bounded. Consequently, each population and virus is bounded. Hence, from \eqref{P_C} we have:

\begin{equation*}\label{din_total}
\begin{array}{lcl}
\dot{P_C} &=& \dot{T} + \dot{I}_H + \dot{I}_S + \dot{C} \\
\\
&=& \lambda - k_1 \left(1-\epsilon_{RT}\right)TV_H - k_2 T V_S - \mu T \\
\\
&& + k_1 \left(1-\epsilon_{RT}\right)TV_H - k_2 I_H V_S - \mu I_H \\ 
\\
&& + k_2 T V_S - k_1\left(1-\epsilon_{RT}\right) I_S V_H - \mu I_S \\
\\
&&+ k_1 \left(1-\epsilon_{RT}\right) I_S V_H + k_2 I_H V_S - \mu C \\
\\
&=& \lambda - \mu (T + I_H + I_S + C) \\
\\
&=& \lambda - \mu P_C.
\end{array}
\end{equation*}

\medskip

Using the classical differential version of the Gronwall's Lemma, we have:

\begin{equation*} 
\begin{array}{lcl}
P_C(t) \leq {P_C}_0 e^{- \mu t} - \dfrac{\lambda}{\mu} (e^{-\mu t} - 1) ,
\end{array}
\end{equation*}

\medskip

from where we conclude that

\begin{equation*}\label{gronwall_final}
\begin{array}{lcl}
\dpt 0 \,\, \leq \,\, \lim_{t \rightarrow \infty}{P_C(t)} \,\, \leq \,\, \lim_{t \rightarrow \infty}{\left[ {P_C}_0 e^{- \mu t} - \dfrac{\lambda}{\mu} (e^{-\mu t} - 1) \right]} \,\, = \,\, \dfrac{\lambda}{\mu} .
\end{array}
\end{equation*}

\medskip

Then, $T$, $I_H$, $I_S$ and $C$ are bounded. From \eqref{P_V}, we write $P_V$ as:

\begin{eqnarray*}
\nonumber \dot{P_V} &=& \dot{V}_H + \dot{V}_S \\
\nonumber \\
&=& n_H \left(1-\epsilon_P\right) \mu I_H - \sigma_H V_H + n_S \mu I_S - \sigma_S V_S. \label{vir_total}
\end{eqnarray*}

\medskip

Let us consider $\sigma_H = \sigma_S = \sigma$. Since $I_H$ and $I_S$ are bounded, then we use $I_H^{\text{max}}$ and $I_S^{\text{max}}$ for the maximum values. Hence, rewriting $\dot{P}_V$ we get

\begin{equation*}\label{vir_total_2}
\begin{array}{lcl}
\dot{P_V} &=& n_H \left(1-\epsilon_P\right) \mu I_H^{\text{max}} + n_S \mu I_S^{\text{max}} - \sigma \left(V_H + V_S \right) \\
\\
&=& n_H \left(1-\epsilon_P\right) \mu I_H^{\text{max}} + n_S \mu I_S^{\text{max}} - \sigma P_V ,
\end{array}
\end{equation*}

\medskip

and making use of the classical differential version of the Gronwall's Lemma, we get

\begin{equation*}\label{gronwall_final_2}
\begin{array}{lcl}
&&\dpt 0 \,\, \leq \,\, \lim_{t \rightarrow \infty}{P_V(t)} \,\, \leq \,\, \lim_{t \rightarrow \infty}{\left[ {P_V}_0 e^{- \sigma t} - \dfrac{n_H \left(1-\epsilon_P\right) \mu I_H^{\text{max}} + n_S \mu I_S^{\text{max}}}{\sigma} (e^{-\sigma t} - 1) \right]} \\
\\
&\Leftrightarrow& \dpt 0 \,\, \leq \,\, \lim_{t \rightarrow \infty}{P_V(t)} \,\, \leq \,\, \dfrac{n_H \left(1-\epsilon_P\right) \mu I_H^{\text{max}} + n_S \mu I_S^{\text{max}}}{\sigma} .
\end{array}
\end{equation*}

Then, $V_H$ and $V_S$ are bounded and the Theorem is proved.
\end{proof}

\noindent We compute the basic reproduction number of HIV and SARS-CoV-2 submodels.

\smallskip

\noindent {\bf Submodel analysis.} In the context of this work, the basic reproduction number $\mathcal{R}_0$ is the number of secondary infections due to a single infected cell in a susceptible healthy cell population \cite{DriesscheWatmough2002}. 

\bigskip

\noindent {\bf HIV submodel.} There is no SARS-CoV-2, imposing $I_S = V_S = C = 0$:

\begin{equation}\label{hiv_submodel}
\begin{array}{lcl}
\dot{x}_{\text{hiv}} = f(x_{\text{hiv}}) \quad \Leftrightarrow \quad
\begin{cases}
&\dot{T} = \lambda - k_1 \left( 1-\epsilon_{RT} \right)TV_H - \mu T \\
\\
&\dot{I}_H = k_1 \left(1-\epsilon_{RT}\right)TV_H - \mu I_H \\
\\
&\dot{V}_H = n_H \left(1-\epsilon_P \right) \mu I_H - \sigma_H V_H \, ,
\end{cases}
\end{array}
\end{equation}

\medskip

\noindent with $x_{\text{hiv}} = (T, I_H, V_H) \in \left(\mathbb{R}_0^+\right)^3$.

\bigskip

\noindent {\bf SARS-CoV-2 submodel.} There is no HIV, imposing $I_H = V_H = C = 0$:

\begin{equation}\label{sarscov_submodel}
\begin{array}{lcl}
\dot{x}_{\text{sars}} = f(x_{\text{sars}}) \quad \Leftrightarrow \quad
\begin{cases}
&\dot{T} = \lambda - k_2 T V_S - \mu T \\
\\
&\dot{I}_S = k_2 T V_S - \mu I_S \\
\\
&\dot{V}_S = n_S \mu I_S - \sigma_S V_S \, ,
\end{cases}
\end{array}
\end{equation}

\medskip

\noindent with $x_{\text{sars}} = (T, I_S, V_S) \in \left(\mathbb{R}_0^+\right)^3$. The vector field associated to \eqref{hiv_submodel} and \eqref{sarscov_submodel} is denoted by $f(x_{\text{hiv}})$ and $f(x_{\text{sars}})$, respectively. Their flows are 

$$
\psi_{\text{hiv}} \left(t, (T_0, {I_H}_0, {V_H}_0) \right)\, , \,\, t \in \mathbb{R}_0^+ \, , \,\, \left(T_0, {I_H}_0, {V_H}_0 \right) \in (\mathbb{R}_0^+)^3
$$

\noindent and

$$
\psi_{\text{sars}} \left(t, (T_0, {I_S}_0, {V_S}_0) \right) \,, \,\, t \in \mathbb{R}_0^+ \,, \,\, \left(T_0, {I_S}_0, {V_S}_0 \right) \in (\mathbb{R}_0^+)^3 \, ,
$$

\medskip

\noindent respectively. Using the same method as in \cite{DriesscheWatmough2002} we compute the basic reproduction number for \eqref{hiv_submodel} and \eqref{sarscov_submodel}. The disease-free equilibrium of system \eqref{hiv_submodel} is given by:

\begin{equation*}\label{DFE_HIV}
\begin{array}{lcl}
E_{\text{hiv}} = \left(T_{\text{hiv}}, {I_H}_{\text{hiv}}, {V_H}_{\text{hiv}} \right) = \left(\dfrac{\lambda}{\mu}, 0, 0 \right) .
\end{array}
\end{equation*}

\medskip

\noindent The matrix $F_{\text{hiv}}$ is the matrix where the entries represent the new HIV infections and the matrix $V_{\text{hiv}}$ represents the remaining terms:

\medskip

\begin{equation*}
\label{FVhiv}
\begin{array}{lcl}
F_{\text{hiv}}=\left(\begin{array}{cc}
0 & k_1 \left(1-\epsilon_{RT} \right) T_{\text{hiv}} \\ 
\\
0 & 0
\end{array}\right) , \quad
V_{\text{hiv}}=\left(\begin{array}{cc}
\mu & 0 \\ 
\\
-n_H \left(1-\epsilon_P \right) \mu & \sigma_H
\end{array}\right) .
\end{array}
\end{equation*}

\medskip

\noindent Then, through \cite[Lemma 1]{DriesscheWatmough2002}, the basic reproduction number associated to model \eqref{hiv_submodel} is given by

\begin{equation}
\label{R_hiv}
\mathcal{R}_{\text{hiv}} = \rho \left(F_{\text{hiv}} V_{\text{hiv}}^{-1}\right) = \dfrac{k_1 n_H  \left( 1-\epsilon_{RT} \right) \left( 1-\epsilon_P \right) \lambda}{\mu \sigma_H},
\end{equation}

\medskip

\noindent where $\rho$ indicates the spectral radius of $F_{\text{hiv}} V_{\text{hiv}}^{-1}$. The disease-free equilibrium of system \eqref{sarscov_submodel} is given by:

\begin{equation*}\label{DFE_SARSCOV}
\begin{array}{lcl}
E_{\text{sars}} = \left(T_{\text{sars}}, {I_S}_{\text{sars}}, {V_S}_{\text{sars}} \right) = \left(\dfrac{\lambda}{\mu}, 0, 0 \right) .
\end{array}
\end{equation*}

\medskip

\noindent The matrix $F_{\text{sars}}$ is the matrix where the entries represent the new SARS-CoV-2 infections and the matrix $V_{\text{sars}}$ represents the remaining terms:

\medskip

\begin{equation*}
\label{FVsars}
\begin{array}{lcl}
F_{\text{sars}}=\left(\begin{array}{cc}
0 & k_2 T_{\text{sars}} \\ 
\\
0 & 0
\end{array}\right) , \quad
V_{\text{sars}}=\left(\begin{array}{cc}
\mu & 0 \\ 
\\
-n_S \mu & \sigma_S
\end{array}\right) .
\end{array}
\end{equation*}

\medskip

\noindent Therefore, through \cite[Lemma 1]{DriesscheWatmough2002}, the basic reproduction number associated to model \eqref{sarscov_submodel} is given by

\begin{equation}
\label{R_sars}
\mathcal{R}_{\text{sars}} = \rho \left(F_{\text{sars}} V_{\text{sars}}^{-1}\right) = \dfrac{k_2 n_S \lambda}{\mu \sigma_S},
\end{equation}

\medskip

\noindent where $\rho$ indicates the spectral radius of $F_{\text{sars}} V_{\text{sars}}^{-1}$.

\smallskip

\noindent {\bf Main results.} Let $E_{\text{hiv}}^e$ and $E_{\text{sars}}^e$ be the endemic equilibria of \eqref{hiv_submodel} and \eqref{sarscov_submodel}, respectively. We set the following results:

\begin{theorem} \label{bifurc}
Let $\mathcal{R}_{\text{hiv}}$ and $\mathcal{R}_{\text{sars}}$ be the basic reproduction numbers of HIV and SARS-CoV-2 submodels, respectively. Therefore,

\medskip

\begin{itemize}
\item[\text{(i)}] $E_{\text{hiv}}^e$ undergoes a transcritical bifurcation at $\mathcal{R}_{\text{hiv}} = 1$, and lies in the interior of the first octant if $\mathcal{R}_{\text{hiv}} > 1$;
    
\bigskip
    
\item[\text{(ii)}] $E_{\text{sars}}^e$ undergoes a transcritical bifurcation at $\mathcal{R}_{\text{sars}} = 1$, and lies in the interior of the first octant if $\mathcal{R}_{\text{sars}} > 1$.
\end{itemize}
\end{theorem}

\medskip

\noindent \noindent Now, let $\mathcal{R}_0$ be the basic reproduction number and $E_{\text{DFE}}$ be the disease-free equilibrium of the full model \eqref{modelo}. 

\begin{theorem} \label{FULL_THEO}
The basic reproduction number of the model \eqref{modelo} is $\mathcal{R}_0 = \max{\{\mathcal{R}_{\text{hiv}}, \mathcal{R}_{\text{sars}}\}}$. Moreover, the disease-free equilibrium point $E_{\text{DFE}}$ is locally asymptotically stable when $\mathcal{R}_0 < 1$. Otherwise it is unstable.
\end{theorem}

\medskip

\noindent In Section \ref{SEC_5} and Section \ref{SEC_6} we present the proof of Theorem \ref{bifurc} and Theorem \ref{FULL_THEO}, respectively.


\section{Sensitivity analysis of $\mathcal{R}_{\text{hiv}}$ and $\mathcal{R}_{\text{sars}}$} \label{SEC_4}

\noindent Sensitivity indices are used to evaluate how a variable varies in response to changes in a given parameter. These indices represent the ratio between the relative change in the variable and the relative change in the parameter. In the case where the variable $v$ is a differentiable function of the parameter $p$, the sensitivity index can be computed using partial derivatives as follows \cite{MauricioPinto2021, chitnis2008}:

\begin{equation*} 
	\begin{array}{lcl}
	\varphi_p^v = \dfrac{\partial v}{\partial p} \times \dfrac{p}{v} .
\end{array}
\end{equation*}

\medskip

\noindent When considering the specific case of the basic reproduction number $\mathcal{R}$, we have:

\begin{equation*} 
\begin{array}{lcl}
\varphi_p^{\mathcal{R}} = \dfrac{\partial \mathcal{R}}{\partial p} \times \dfrac{p}{\mathcal{R}} .
\end{array}
\end{equation*}

\medskip

\noindent We compute the signs of the sensitivity indices related to $\mathcal{R}_{\text{hiv}}$ and $\mathcal{R}_{\text{sars}}$ and we present them in Table \ref{indices}.

\begin{table}[!h]
\centering
\def\arraystretch{2.5}  
\begin{tabular}{c c @{\hspace{15pt}} c c}
\hline
\textbf{Index} & \textbf{Sensitivity index sign} $\left(\mathcal{R}_{\text{hiv}}\right)$
               & \textbf{Index} & \textbf{Sensitivity index sign} $\left(\mathcal{R}_{\text{sars}}\right)$ \\
                             
\hline

$\varphi_{k_1}^{\mathcal{R}_{\text{hiv}}}$ & $+1$ & $\varphi_{k_2}^{\mathcal{R}_{\text{sars}}}$ & $+ 1$ \\

$\varphi_{n_H}^{\mathcal{R}_{\text{hiv}}}$ & $+1$ & $\varphi_{n_S}^{\mathcal{R}_{\text{sars}}}$ & $+1$ \\

$\varphi_{\epsilon_{RT}}^{\mathcal{R}_{\text{hiv}}}$ & $-\dfrac{\epsilon_{RT}}{1 - \epsilon_{RT}} \, < \, 0$ &  &  \\ 

$\varphi_{\epsilon_P}^{\mathcal{R}_{\text{hiv}}}$ & $-\dfrac{\epsilon_P}{1 - \epsilon_P} \, < \, 0$ &  &  \\

$\varphi_{\lambda}^{\mathcal{R}_{\text{hiv}}} $ & $+1$ & $\varphi_{\lambda}^{\mathcal{R}_{\text{sars}}}$ & $+1$ \\

$\varphi_{\mu}^{\mathcal{R}_{\text{hiv}}} $ & $-1$ & $\varphi_{\mu}^{\mathcal{R}_{\text{sars}}}$ & $-1$ \\

$\varphi_{\sigma_H}^{\mathcal{R}_{\text{hiv}}} $ & $-1$ & $\varphi_{\sigma_S}^{\mathcal{R}_{\text{sars}}}$ & $-1$ \\
\hline
\end{tabular}
\medskip
\medskip
\caption{Sensitivity indices for the parameters of $\mathcal{R}_{\text{hiv}}$ and $\mathcal{R}_{\text{sars}}$.}
\label{indices}
\end{table}

\noindent The transmission rates $k_1$ and $k_2$ contribute to the increase in the basic reproduction number of HIV and SARS-CoV-2, respectively. The parameters $n_H$, $n_S$ and $\lambda$ have the same effect on the respective basic reproduction numbers. On the other hand, the parameters relating to treatment and mortality rates of infected cells and viruses have the opposite effect, {\it  i.e.} increasing the values of these parameters has a retarding effect on the basic reproduction number of the respective viruses.


\section{Proof of Theorem \ref{bifurc}} \label{SEC_5}

\medskip

\noindent With respect to the HIV submodel \eqref{hiv_submodel}, through Theorem 2 of \cite{DriesscheWatmough2002}, we obtain the following Lemma:

\begin{lemma} \label{hiv_lemma}
If \,$\mathcal{R}_{\text{hiv}} < 1$, then the disease-free equilibrium point $E_{\text{hiv}}$ of \eqref{hiv_submodel} is locally asymptotically stable. Otherwise it is unstable.
\end{lemma}

\begin{proof}
Let

\begin{eqnarray}
\nonumber \mathcal{J}_{\text{hiv}}&=&\left(\begin{array}{ccc}
-k_1  \left(1-\epsilon_{RT} \right) {V_H}_{\text{hiv}} - \mu & 0 & -k_1  \left(1-\epsilon_{RT} \right) T_{\text{hiv}}  \\ 
\nonumber \\
k_1  \left(1-\epsilon_{RT} \right) {V_H}_{\text{hiv}} & -\mu & k_1  \left(1-\epsilon_{RT} \right) T_{\text{hiv}} \\
\nonumber \\
0 & n_H \left(1-\epsilon_P \right) \mu & -\sigma_H
\end{array}\right) \\
\nonumber \\
\nonumber \\
\nonumber \\
&=& \left(\begin{array}{ccc}
- \mu & 0 & -\dfrac{k_1  \left(1-\epsilon_{RT} \right) \lambda}{\mu}  \\ 
\\
0 & -\mu & \dfrac{k_1  \left(1-\epsilon_{RT} \right) \lambda}{\mu}  \\
\\
0 & n_H \left(1-\epsilon_P \right) \mu & -\sigma_H \label{jaco_hiv}
\end{array}\right)
\end{eqnarray}

\medskip

be the matrix of linearization of model \eqref{hiv_submodel} around $E_{\text{hiv}}$. Then, the associated eigenvalues are: 

\begin{eqnarray*}
\label{eigen_hiv}
\varphi_1 &=& -\mu \\
\\
\varphi_2 &=& \dfrac{-\left(\sigma_H + \mu\right) + \sqrt{4n_H\lambda\left(1-\epsilon_{RT} \right)\left(1-\epsilon_P \right) k_1 + (\mu - \sigma_H)^2}}{2} \\
\\
\varphi_3 &=& \dfrac{-\left(\sigma_H + \mu\right) - \sqrt{4n_H\lambda\left(1-\epsilon_{RT} \right)\left(1-\epsilon_P \right) k_1 + \left(\mu - \sigma_H\right)^2}}{2} .
\end{eqnarray*}

\medskip

It is trivial to see that the eigenvalues $\varphi_1$ and $\varphi_3$ have negative real part. Regarding to the eigenvalue $\varphi_2$, it has negative real part if

\begin{eqnarray*}
\label{eigen_hiv}
&& -\left(\sigma_H + \mu\right) + \sqrt{4n_H\lambda\left(1-\epsilon_{RT} \right)\left(1-\epsilon_P \right) k_1 + (\mu - \sigma_H)^2}< 0 \\
\\
& \Leftrightarrow& 4n_H\lambda \left(1-\epsilon_{RT} \right)\left(1-\epsilon_P \right) k_1 + (\mu - \sigma_H)^2 <  \left(\sigma_H + \mu\right)^2  \\
\\
& \Leftrightarrow& k_1  n_H \left(1-\epsilon_{RT} \right)\left(1-\epsilon_P \right) \lambda < \mu \sigma_H \\
\\
&\overset{\eqref{R_hiv}}{ \Leftrightarrow}& \mathcal{R}_{\text{hiv}} < 1.
\end{eqnarray*}

\medskip

Hence, $\varphi_2 < 0$ and all eigenvalues have negative real part if $\mathcal{R}_{\text{hiv}} < 1$. This proves Lemma \ref{hiv_lemma}.

\end{proof}

\noindent Now, computing $E_{\text{hiv}}^e$, we get:

\begin{equation*}\label{END_HIV}
\begin{array}{lcl}
E^e_{\text{hiv}} &=& \left(T^e_{\text{hiv}}, {I_H}^e_{\text{hiv}}, {V_H}_{\text{hiv}}^e \right) \, ,
\end{array}
\end{equation*}

\noindent where

\begin{equation*}\label{END_HIV_2}
\begin{array}{lcl}
T_{\text{hiv}}^e & = & \dfrac{\sigma_H}{n_H k_1 \left( 1 - \epsilon_{RT} \right) \left( 1 - \epsilon_P \right)} \\
\\
\\
{I_H}_{\text{hiv}}^e & = & \dfrac{n_H k_1 \lambda \left( 1 - \epsilon_{RT} \right) \left( 1 - \epsilon_P \right) - \mu \sigma_H }{n_H \mu k_1  \left( 1 - \epsilon_{RT} \right) \left( 1 - \epsilon_P \right)} \\
\\
\\
{V_H}_{\text{hiv}}^e & = & \dfrac{n_H k_1 \lambda \left( 1 - \epsilon_{RT} \right) \left( 1 - \epsilon_P \right) - \mu \sigma_H }{k_1 \sigma_H \left( 1 - \epsilon_{RT} \right)}.
\end{array}
\end{equation*}

\medskip

\noindent $E_{\text{hiv}}^e$ lies in the interior of the first octant if $T_{\text{hiv}}^e$, ${I_H}_{\text{hiv}}^e$ and ${V_H}_{\text{hiv}}^e$ are positive. It is clear that $T_{\text{hiv}}^e > 0$. Since $n_H k_1 \lambda \left( 1 - \epsilon_{RT} \right) \left( 1 - \epsilon_P \right) - \mu \sigma_H > 0$, we have ${I_H}_{\text{hiv}}^e > 0$ and ${V_H}_{\text{hiv}}^e > 0$. Hence,

\begin{equation*}\label{LIES_HIV}
\begin{array}{lcl}
&& n_H k_1 \lambda \left( 1 - \epsilon_{RT} \right) \left( 1 - \epsilon_P \right) - \mu \sigma_H > 0 \\
\\
\\
&\Leftrightarrow& \dfrac{n_H k_1 \lambda \left( 1 - \epsilon_{RT} \right) \left( 1 - \epsilon_P \right)}{\mu \sigma_H} > 1 \\
\\
&\overset{\eqref{R_hiv}}{\Leftrightarrow}& \mathcal{R}_{\text{hiv}} > 1 \, ,
\end{array}
\end{equation*}

\medskip

\noindent and $E_{\text{hiv}}^e$ lies in the interior of the first octant. Now, we apply the same process as in Lemma \ref{hiv_lemma} to analyze the stability of $E_{\text{hiv}}^e$:

\begin{lemma} \label{hiv_lemma_end}
If \,$\mathcal{R}_{\text{hiv}} > 1$, then the endemic equilibrium point $E_{\text{hiv}}^e$ of \eqref{hiv_submodel} is locally asymptotically stable. Otherwise it is unstable.
\end{lemma}

\begin{proof}
Let

\begin{equation*}
\label{jaco_hiv_end}
\begin{array}{lcl}
\mathcal{J}_{\text{hiv}}^e&=& \left(\begin{array}{ccc}
-\dfrac{n_H \lambda k_1  \left( 1 - \epsilon_{RT} \right) \left( 1 - \epsilon_P \right)}{\sigma_H} & 0 & -\dfrac{\sigma_H}{n_H \left(1-\epsilon_P \right)}  \\ 
\\
\dfrac{n_H \lambda k_1  \left( 1 - \epsilon_{RT} \right) \left( 1 - \epsilon_P \right)  -\mu \sigma_H}{\sigma_H} & -\mu & \dfrac{\sigma_H}{n_H \left(1-\epsilon_P \right)}  \\
\\
0 & n_H \left(1-\epsilon_P \right) \mu & -\sigma_H
\end{array}\right)
\end{array}
\end{equation*}

\medskip

be the matrix of linearization of model \eqref{hiv_submodel} around $E_{\text{hiv}}^e$. Then, the associated eigenvalues are: 

\begin{eqnarray*}
\label{eigen_hiv_end}
\varphi_1^e &=& -\mu \\
\\
\varphi_2^e &=& - \dfrac{ k_1 \lambda n_H  \left(1-\epsilon_{RT} \right)\left(1-\epsilon_P \right) + \sigma_H^2}{2\sigma_H} \\
\\
&&+ \dfrac{\sqrt{n_H^2 k_1^2  \left(1-\epsilon_{RT} \right)^2 \left(1-\epsilon_P \right)^2 \lambda^2 - 2 n_H k_1 \sigma_H^2  \left(1-\epsilon_{RT} \right)\left(1-\epsilon_P \right) \lambda + 4 \mu \sigma_H^3 +\sigma_H^4}}{2\sigma_H} \\
\\
\varphi_3^e &=&  - \dfrac{ k_1 \lambda n_H  \left(1-\epsilon_{RT} \right)\left(1-\epsilon_P \right) + \sigma_H^2}{2\sigma_H} \\
\\
 &&- \dfrac{\sqrt{n_H^2 k_1^2  \left(1-\epsilon_{RT} \right)^2 \left(1-\epsilon_P \right)^2 \lambda^2 - 2 n_H k_1 \sigma_H^2  \left(1-\epsilon_{RT} \right)\left(1-\epsilon_P \right) \lambda + 4 \mu \sigma_H^3 +\sigma_H^4}}{2\sigma_H} .
\end{eqnarray*}

\medskip

It is trivial to see that the eigenvalues $\varphi_1^e$ and $\varphi_3^e$ have negative real part. Repeating the algebraic manipulations as in Lemma \ref{hiv_lemma}, the eigenvalue $\varphi_2^e$ has negative real part if

$$
\varphi_2^e < 1 \,\, \Leftrightarrow \,\, \mathcal{R}_{\text{hiv}} > 1 \, .
$$

\medskip

and consequently $E_{\text{hiv}}^e$ is stable. Otherwise is unstable. Hence, Lemma \ref{hiv_lemma_end} is proved.
\end{proof}

\medskip

\noindent Hence, $E_{\text{hiv}}^e$ undergoes a trancritical bifurcation at $\mathcal{R}_{\text{hiv}} = 1$ and interchanges its stability with $E_{\text{hiv}}$.

\medskip

\noindent We repeat the same procedure for the SARS-CoV-2 submodel as in the HIV submodel \eqref{sarscov_submodel}. Then, through Theorem 2 of \cite{DriesscheWatmough2002}, we obtain the following lemma:

\begin{lemma} \label{sars_lemma}
If \,$\mathcal{R}_{\text{sars}} < 1$, then the disease-free equilibrium point $E_{\text{sars}}$ of \eqref{sarscov_submodel} is locally asymptotically stable. Otherwise it is unstable.
\end{lemma}

\begin{proof}
Let

\begin{equation}
\label{jaco_sars}
\begin{array}{lcl}
\mathcal{J}_{\text{sars}}=\left(\begin{array}{ccc}
-k_2  {V_S}_{\text{sars}} - \mu & 0 & -k_2  T_{\text{sars}}  \\ 
\\
k_2  {V_S}_{\text{sars}} & -\mu & k_2  T_{\text{sars}} \\
\\
0 & n_S \mu & -\sigma_S
\end{array}\right) = \left(\begin{array}{ccc}
- \mu & 0 & -\dfrac{k_2 \lambda}{\mu}  \\ 
\\
0 & -\mu & \dfrac{k_2 \lambda}{\mu}  \\
\\
0 & n_S \mu & -\sigma_S
\end{array}\right)
\end{array}
\end{equation}

\medskip

be the matrix of linearization of model \eqref{sarscov_submodel} around $E_{\text{sars}}$. Then, the associated eigenvalues are: 

\begin{eqnarray*}
\label{eigen_hiv}
\varphi_4 &=& -\mu \\
\\
\varphi_5 &=& \dfrac{-\left(\sigma_S + \mu \right) + \sqrt{4 n_S \lambda  k_2 + \left( \mu - \sigma_S \right)^2}}{2} \\
\\
\varphi_6 &=& \dfrac{-\left(\sigma_S + \mu \right) - \sqrt{4 n_S \lambda  k_2 + \left( \mu - \sigma_S \right)^2}}{2} .
\end{eqnarray*}

\medskip

Analogously to what was done in the HIV submodel, it is easy to see that the eigenvalues $\varphi_4$ and $\varphi_6$ have negative real part. Hence, if

\begin{eqnarray*}
\label{eigen_hiv}
&& -\left(\sigma_S + \mu \right) + \sqrt{4 n_S \lambda  k_2 + \mu^2 - 2\mu \sigma_S + {\sigma_S}^2}< 0 \\
\\
& \Leftrightarrow& 4 n_S \lambda  k_2 + \mu^2 - 2\mu \sigma_S + {\sigma_S}^2 < \left(\sigma_S + \mu \right)^2  \\
\\
& \Leftrightarrow& k_1  n_H \left(1-\epsilon_{RT} \right)\left(1-\epsilon_P \right) \lambda < \mu \sigma_H \\
\\
& \Leftrightarrow& k_2 n_S \lambda < \mu \sigma_S \\
\\
&\overset{\eqref{R_sars}}{ \Leftrightarrow}& \mathcal{R}_{\text{sars}} < 1,
\end{eqnarray*}

\medskip

then $\varphi_5 < 0$ and all eigenvalues have negative real part. This proves Lemma \ref{sars_lemma}.
\end{proof}

\noindent Now, we compute the endemic equilibrium point for SARS-CoV-2 submodel \eqref{sarscov_submodel}. We get:

\begin{equation*}\label{END_SARS}
\begin{array}{lcl}
\medskip
E^e_{\text{sars}} &=& \left(T^e_{\text{sars}}, {I_S}^e_{\text{sars}}, {V_S}_{\text{sars}}^e \right) \\
\\
\medskip
&=& \left(\dfrac{\sigma_S}{n_S k_2 }, \dfrac{n_S k_2 \lambda - \mu \sigma_S }{n_S \mu k_2}, \dfrac{n_S k_2 \lambda - \mu \sigma_S}{k_2 \sigma_S} \right) .
\end{array}
\end{equation*}

%

\medskip

\noindent $E_{\text{sars}}^e$ lies in the interior of the first octant if $T_{\text{sars}}^e$, ${I_S}_{\text{sars}}^e$ and ${V_S}_{\text{sars}}^e$ are positive. It is clear that $T_{\text{sars}}^e > 0$. Since $n_S k_2 \lambda - \mu \sigma_S  > 0$, we have ${I_S}_{\text{sars}}^e > 0$ and ${V_S}_{\text{sars}}^e > 0$. Hence,

\begin{equation*}\label{LIES_SARS}
\begin{array}{lcl}
\smallskip
&& n_S k_2 \lambda - \mu \sigma_S  > 0 \\
\\
&\Leftrightarrow& \dfrac{n_S k_2 \lambda}{\mu \sigma_S} > 1 \\
\\
&\overset{\eqref{R_sars}}{\Leftrightarrow}& \mathcal{R}_{\text{sars}} > 1 \, ,
\end{array}
\end{equation*}

\medskip

\noindent and $E_{\text{sars}}^e$ lies in the interior of the first octant. Now, we apply the same process as in Lemma \ref{sars_lemma} to analyze the stability of $E_{\text{sars}}^e$:

 \begin{lemma} \label{sars_lemma_end}
If \,$\mathcal{R}_{\text{sars}} > 1$, then the endemic equilibrium point $E_{\text{sars}}^e$ of \eqref{sarscov_submodel} is locally asymptotically stable. Otherwise it is unstable.
\end{lemma}

\begin{proof}
Let

\begin{equation*}
\label{jaco_sars_end}
\begin{array}{lcl}
\mathcal{J}_{\text{sars}}^e&=& \left(\begin{array}{ccc}
-\dfrac{n_H \lambda k_2  }{\sigma_S} & 0 & -\dfrac{\sigma_S}{n_S}  \\ 
\\
\dfrac{n_H \lambda k_2  -\mu \sigma_S}{\sigma_S} & -\mu & \dfrac{\sigma_S}{n_S}  \\
\\
0 & n_S \mu & -\sigma_S
\end{array}\right)
\end{array}
\end{equation*}

\medskip

be the matrix of linearization of model \eqref{sarscov_submodel} around $E_{\text{sars}}^e$. Then, the associated eigenvalues are: 

\begin{eqnarray*}
\label{eigen_sars_end}
\varphi_4^e &=& -\mu \\
\\
\varphi_5^e &=& - \dfrac{ k_2 \lambda n_S + \sigma_S^2 - \sqrt{n_S^2 k_2^2 \lambda^2 - 2 n_S k_2 \sigma_S^2 \lambda + 4 \mu \sigma_S^3 +\sigma_S^4}}{2\sigma_S} \\
\\
\varphi_6^e &=&  - \dfrac{ k_2 \lambda n_S + \sigma_S^2 + \sqrt{n_S^2 k_2^2 \lambda^2 - 2 n_S k_2 \sigma_S^2 \lambda + 4 \mu \sigma_S^3 +\sigma_S^4}}{2\sigma_S} .
\end{eqnarray*}

\medskip

It is trivial to see that the eigenvalues $\varphi_4^e$ and $\varphi_6^e$ have negative real part. Repeating the algebraic manipulations as in Lemma \ref{hiv_lemma}, the eigenvalue $\varphi_5^e$ has negative real part if

$$
\varphi_5^e < 1 \,\, \Leftrightarrow \,\, \mathcal{R}_{\text{sars}} > 1 \, .
$$

\medskip

and consequently $E_{\text{sars}}^e$ is stable. Otherwise is unstable. Hence, Lemma \ref{sars_lemma_end} is proved.
\end{proof}

\section{Proof of Theorem \ref{FULL_THEO}} \label{SEC_6}

\noindent We analyze the full model \eqref{modelo}. The disease-free equilibrium point of system \eqref{modelo} is given by:

\begin{equation*}\label{DFE_FULL}
\begin{array}{lcl}
E_{\text{DFE}} = \left(T^{\star}, {I_H}^{\star}, {I_S}^{\star}, {V_H}^{\star}, {V_S}^{\star}, C^{\star} \right) = \left(\dfrac{\lambda}{\mu}, 0, 0, 0, 0, 0 \right) .
\end{array}
\end{equation*}

\noindent The proof of the first result of this theorem follows from the result obtained by the authors of \cite{DriesscheWatmough2002}, {\it i.e.} the basic reproduction number of a model with multiples infections in interaction can be approximated as the maximum of the basic reproduction number of each submodel. This result is valid because, in the long-term dynamics of the system, the infection with the highest basic reproduction number controls the spread of the disease, which determines whether or not the infection can persist in the population. Accordingly, applying this principle to \eqref{modelo}, the basic reproduction number is given by

\begin{equation*}\label{R0_FULL}
\begin{array}{lcl}
\mathcal{R}_0 = \max{\{\mathcal{R}_{\text{hiv}}, \mathcal{R}_{\text{sars}}\}} \,.
\end{array}
\end{equation*}

\noindent Now, let

\begin{equation*}
\label{jaco_coinf}
\begin{array}{lcl}
\mathcal{J}_{\text{f}} &=& \adjustbox{max width=\textwidth}{
$\left(\begin{array}{cccccc}
- k_1 \left(1 - \epsilon_{RT} \right) {V_H}^{\star} - k_2 {V_S}^{\star} - \mu & 0 & 0 & - k_1 \left(1 - \epsilon_{RT} \right) T^{\star} & -k_2 T^{\star} & 0  \\ 
\\
k_1 \left(1 - \epsilon_{RT} \right) {V_H}^{\star} & - \left(k_2 {V_S}^{\star} + \mu \right)& 0 & k_1 \left(1 - \epsilon_{RT} \right) T^{\star} & - k_2 {I_H}^{\star} & 0 \\
\\
k_2 {V_S}^{\star} & 0 & - \left( k_1 \left(1 - \epsilon_{RT} \right) {V_H}^{\star} + \mu \right) & - k_1 \left(1 - \epsilon_{RT} \right) {I_S}^{\star} & k_2 T^{\star} & 0 \\
\\
0 & n_H \mu \left(1 - \epsilon_P \right) & 0 & -\sigma_H & 0 & 0 \\
\\
0 & 0 & n_S \mu & 0 & - \sigma_S & 0 \\
\\
0 & k_2 {V_S}^{\star} & k_1 \left(1 - \epsilon_{RT} \right) {V_H}^{\star} & - k_1 \left(1 - \epsilon_{RT} \right) {I_S}^{\star} & k_2 {I_H}^{\star} & - \mu
\end{array}\right)$} \\
\\
&=& \adjustbox{max width=\textwidth}{
$\left(\begin{array}{cccccc}
- \mu & 0 & 0 &- \dfrac{k_1 \left(1 - \epsilon_{RT} \right) \lambda}{\mu} & - \dfrac{k_2 \lambda}{\mu} & 0 \\ 
\\
0 & - \mu & 0 &\dfrac{k_1 \left(1 - \epsilon_{RT} \right) \lambda}{\mu} & 0 & 0  \\ 
\\
0 & 0 & - \mu & 0 & \dfrac{k_2 \lambda}{\mu} & 0  \\ 
\\
0 & n_H  \left(1 - \epsilon_P \right) \mu & 0 &- \sigma_H & 0 & 0  \\ 
\\
0 & 0 & n_S \mu &0 & -\sigma_S & 0  \\ 
\\
0 & 0 & 0 & 0 & 0 & - \mu 
\end{array}\right) $}
\end{array}
\end{equation*}

\medskip

\noindent be the matrix of linearization of \eqref{modelo} around $E_{\text{DFE}}$. Then, the associated eigenvalues are

$$
\varphi_1 \, , \quad \varphi_2 \, , \quad \varphi_3 \, , \quad \varphi_4 \, , \quad \varphi_5 \quad \text{and} \quad \varphi_6 \, ,
$$

\medskip

\noindent the same eigenvalues of \eqref{jaco_hiv} and \eqref{jaco_sars} combined. We already know that $\varphi_1$, $\varphi_3$, $\varphi_4$ and $\varphi_6$ have negative real part and $\varphi_2$ and $\varphi_5$ have negative real part if $\mathcal{R}_{\text{hiv}}$ and $\mathcal{R}_{\text{sars}}$ are less than one. Therefore, $E_{\text{DFE}}$ is asymptotically stable if $\mathcal{R}_0 < 1$ and Theorem \ref{FULL_THEO} is proved.

\newpage

\section{Numerics} \label{SEC_7}

\noindent In this section, we present several numerical simulations to analyze the dynamics of HIV and SARS-CoV-2 coinfection. By varying the efficacy of $\epsilon_{RT}$ and $\epsilon_P$, we explore how these treatments impact infected cells and viral loads over time. We use the parameter values given in Table \ref{tabela} for all figures.

\begin{table}[ht!]
\centering
\scalebox{0.89}{
\def\arraystretch{1.6}
\begin{tabular}{ lclc }
\hline\noalign{\smallskip}
\textbf{Parameter} & \textbf{Symbol} & \textbf{Value} & \textbf{Reference}  \\
\noalign{\smallskip}\hline\noalign{\smallskip}
Constant production rate of $T$ cells & $\lambda$ & $10$ cells mm$^{-3}$  & \cite{BairagiAdak2017}   \\
HIV infection rate & $k_1$ & $10^{-8}$ virions mm$^3$ day$^{-1}$  & \cite{CarvalhoPinto2022} \\
SARS-CoV-2 infection rate & $k_2$ & $10^{-3}$ virions mm$^3$ day$^{-1}$  & \cite{Tang2017} \\
RTI-based treatment efficacy & $\epsilon_{RT}$ & [0,1]  & ------ \\
PI drug efficacy & $\epsilon_P$ & [0,1]  & ------ \\
Bursting size for HIV growth & $n_H$ & $42 - 88$ virions cell$^{-1}$ & \cite{RoyBairagi2009} \\
Bursting size for SARS-CoV-2 growth & $n_S$ & $10 - 2500$ virions cell$^{-1}$ & \cite{BairagiAdak2017} \\
Natural death rate of $T$, $I_H$, $I_S$ and $C$ cells & $\mu$ & $10^{-2}$ day$^{-1}$ & \cite{BairagiAdak2017}  \\
Death rate of HIV & $\sigma_H$ & $2 - 3$ day$^{-1}$  & \cite{RoyBairagi2009} \\
Death rate of SARS-CoV-2 & $\sigma_S$ & $3$ day$^{-1}$  & \cite{BairagiAdak2017} \\
\noalign{\smallskip}\hline
\end{tabular}
}
\smallskip
\caption{\label{tabela}Parameter values used in numerical simulations.}
\end{table}

\noindent Figure \ref{hiv_submodel_figure} shows the dynamics of HIV-infected cells $I_H$ and HIV viral load $V_H$ of \eqref{hiv_submodel} for different values of $\epsilon_{RT}$ and $\epsilon_P$, representing the effectiveness of RTI and PI. A noticeable decline in both $I_H$ and $V_H$ populations is observed over time for all cases, with a more rapid and effective reduction as the treatment efficacy parameters increase. We also observe that when treatment efficacy is high enough, the number of HIV-infected cells is higher than the HIV viral load. Figure \ref{sarscov_submodel_figure} illustrates the dynamics of $I_S$ and $V_S$ over time, for the submodel \eqref{sarscov_submodel}. The dashed curve, representing $I_S$, increases rapidly until it stabilizes, while the continuous curve, representing $V_S$, remains practically constant and much lower than $I_S$. However, it can be seen that at the inflection point of the curve of $I_S$, there is also a slight increase in viral concentration, so this seems to be a coherent conclusion and in line with the simulation observation.  Figure \ref{coinfection_figure} shows the dynamics of $I_H$, $I_S$ and coinfected cells with both viruses $C$ over time, for different treatment efficacy values $\epsilon_{RT}$ and $\epsilon_P$:

\begin{itemize}
\item In the first column ``without treatment'', the number of HIV-infected cells shows an oscillatory behavior before decreasing. Initially, the number of cells infected with SARS-CoV-2 increases significantly. Asymptotically, the number of these cells tends to zero. The number of coinfected cells increases exponentially;

\smallskip

\item With increasing treatment efficacy $\epsilon_{RT} = \epsilon_P$ from $40\%$ to $80\%$, a faster decrease in $I_H$ cells is observed. $I_S$ cells continue to grow, but the growth rate decreases as treatment efficacy increases. The number of coinfected cells tends to decrease as the effectiveness of the treatments increases.
\end{itemize}

\noindent Figure \ref{coinfection_figure_2} shows the dynamics of $V_H$ and $V_S$ over time, for different values of $\epsilon_{RT}$ and $\epsilon_P$:

\begin{itemize}
\item In the first column, ``Without treatment'', the HIV viral load increases in an approximately linear fashion. On the other hand, the concentration of SARS-CoV-2 increases quickly, reaches a peak and then decreases until it stabilizes near zero;

\smallskip

\item As treatment efficacy, $\epsilon_{RT} = \epsilon_P$, increases from 40\% to 80\%, two scenarios occur: (i) $V_H$ begins to be controlled more efficiently, leading to a decrease in the viral load; (ii) For any value of the effectiveness of the treatment, $V_S$ has a similar asymptotic behavior, increasing rapidly at first and then stabilizing this rate of increase. However, this growth is reduced as $\epsilon_{RT}$ and $\epsilon_P$ increase.
\end{itemize}

\noindent The simulation in Figure \ref{epsilons_VS} shows that antiretroviral therapy has an influence on the dynamics of SARS-CoV-2 infected cells when in an environment where there are also HIV-infected cells. It reveals that the higher the $\epsilon_{RT}$ and the lower the $\epsilon_P$, the slower the growth of $I_S$ and $V_S$, even though both grow continuously.

\begin{figure}[ht!]
\includegraphics[width=1.155\textwidth]{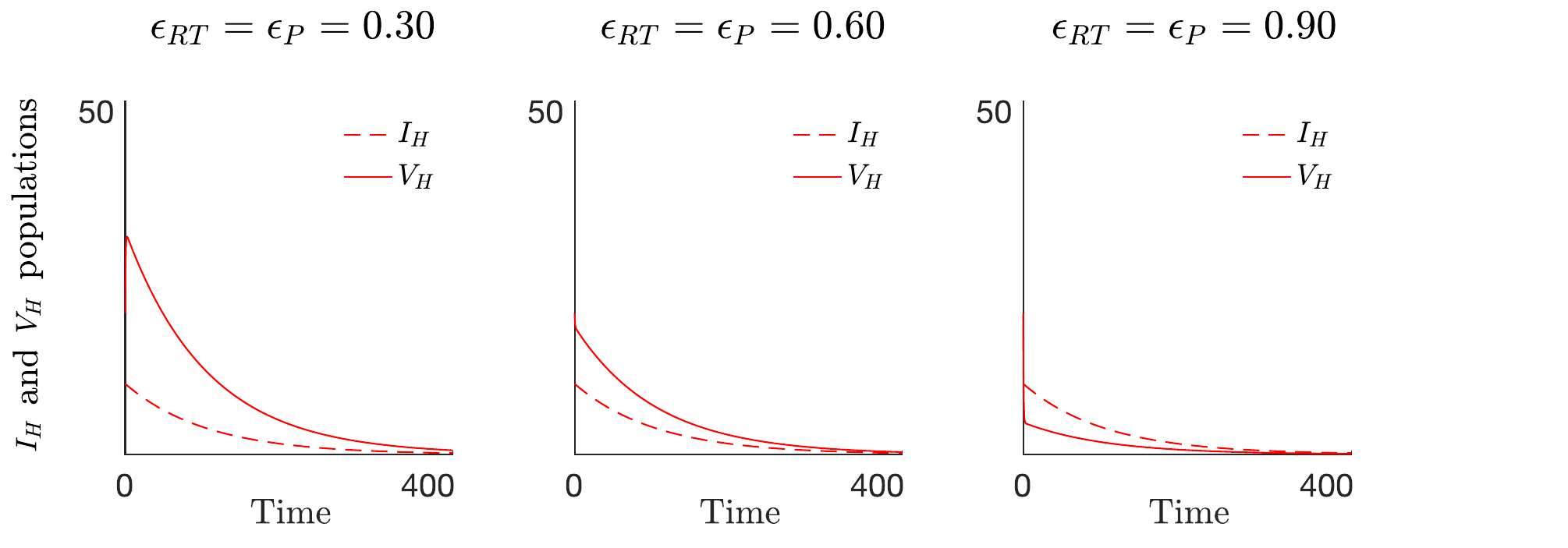}
\caption{$I_H$ and $V_H$ dynamics of system \eqref{hiv_submodel} for different values of $\epsilon_{RT}$ and $\epsilon_P$. Initial conditions: $(T, I_H, V_H) = (10, 10, 20)$.}
\label{hiv_submodel_figure}
\end{figure}

\begin{figure}[ht!]
\includegraphics[width=1.10\textwidth]{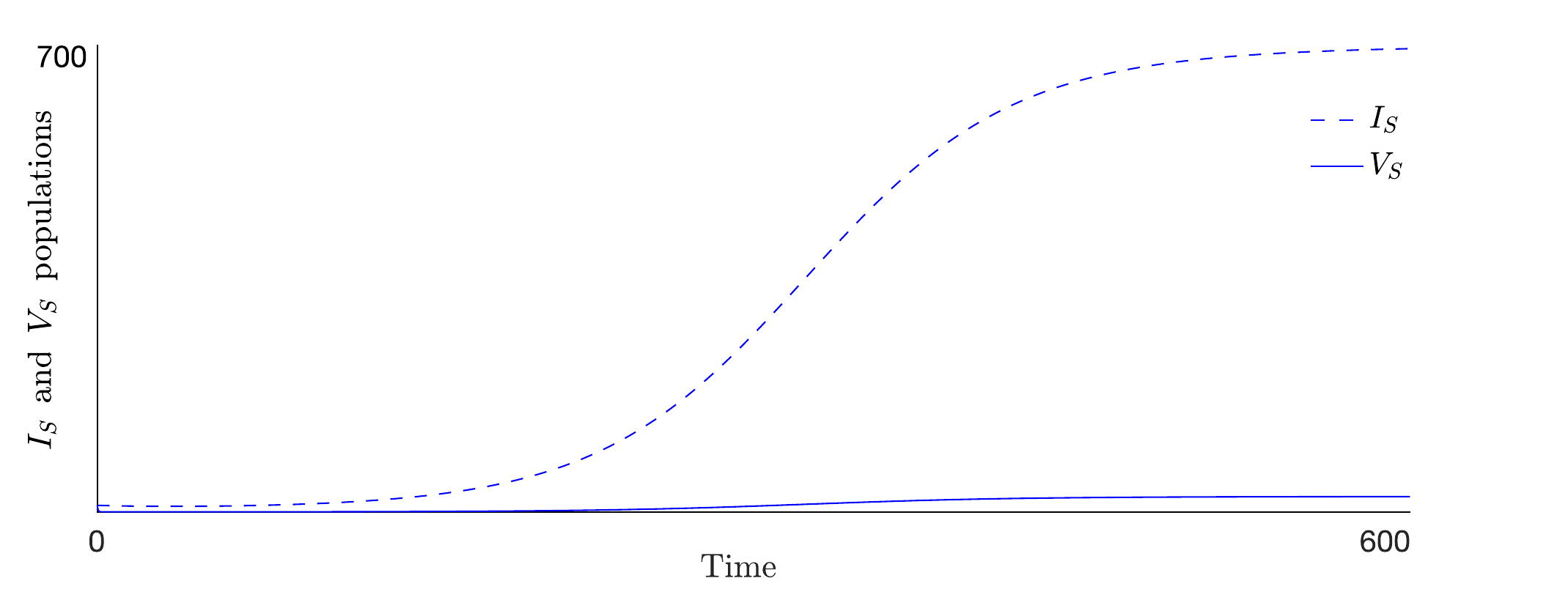}
\caption{$I_S$ and $V_S$ dynamics of system \eqref{sarscov_submodel}. Initial conditions: $(T, I_S, V_S) = (10, 10, 10)$.}
\label{sarscov_submodel_figure}
\end{figure}

\begin{figure}[ht!]
\includegraphics[width=1.08\textwidth]{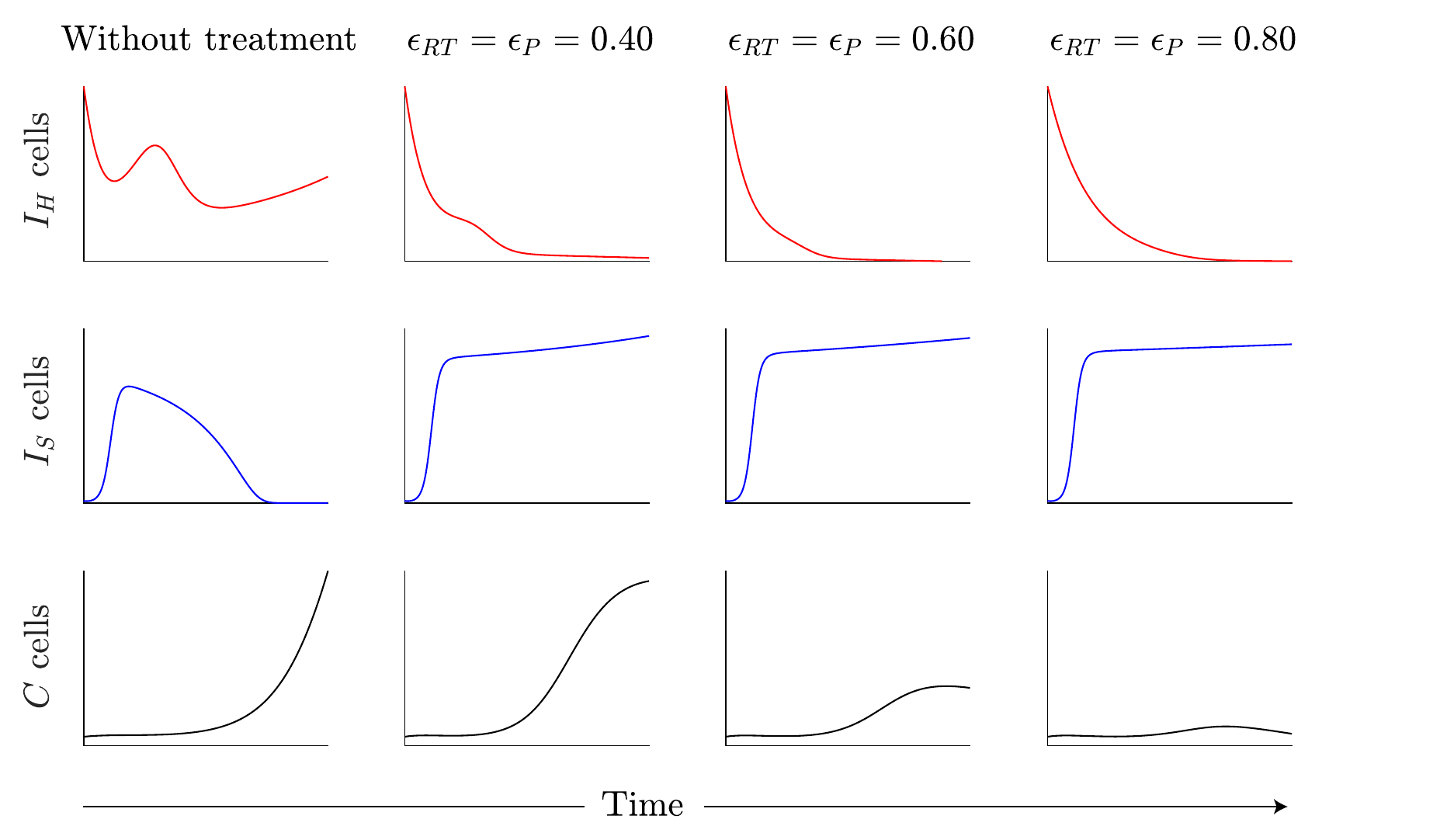}
\caption{$I_H$, $I_S$ and $C$ dynamics of system \eqref{modelo} for different values of $\epsilon_{RT}$ and $\epsilon_P$. Initial conditions: $(T, I_H, I_S, V_H, V_S, C) = (10, 60, 10, 50, 1, 1)$.}
\label{coinfection_figure}
\end{figure}

\begin{figure}[ht!]
\includegraphics[width=1.08\textwidth]{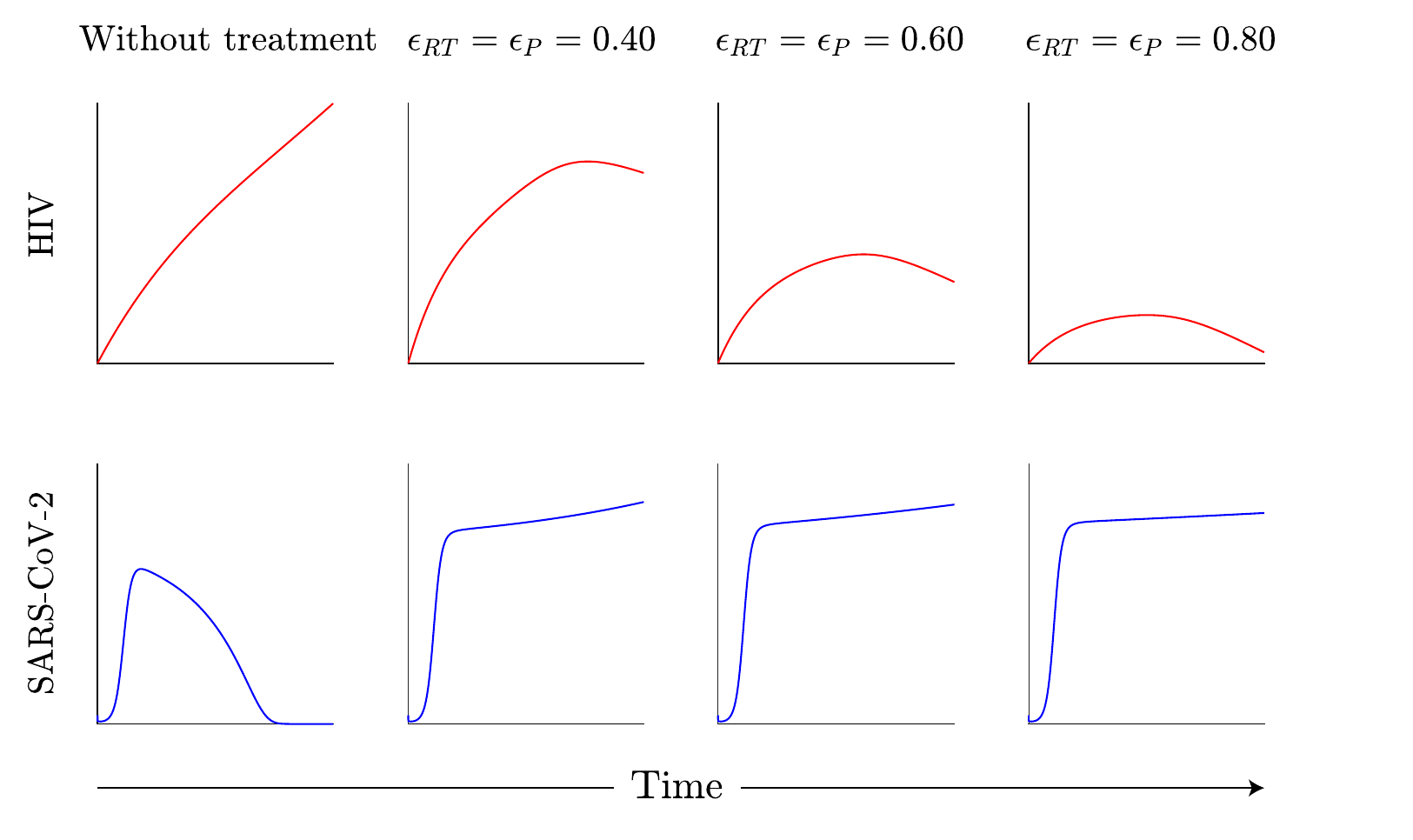}
\caption{$V_H$ and $V_S$ dynamics of system \eqref{modelo} for different values of $\epsilon_{RT}$ and $\epsilon_P$. Initial conditions: $(T, I_H, I_S, V_H, V_S, C) = (10, 60, 10, 50, 1, 1)$.}
\label{coinfection_figure_2}
\end{figure}

\begin{figure}[ht!]
\includegraphics[width=1.1\textwidth]{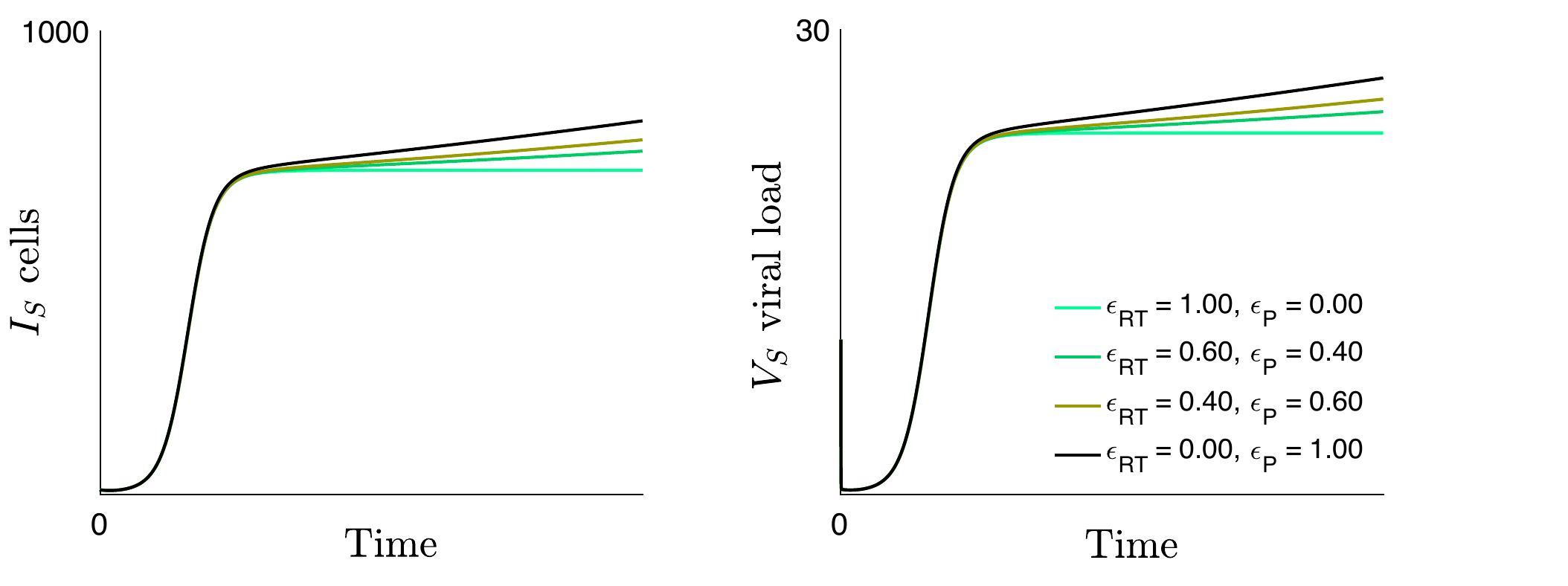}
\caption{$I_S$ and $V_S$ dynamics of system \eqref{modelo} and under HAART therapy with parameter combination $\epsilon_{RT} + \epsilon_P = 1$. Initial conditions: $(T, I_H, I_S, V_H, V_S, C) = (10, 10, 10, 10, 10, 1)$.}
\label{epsilons_VS}
\end{figure}


\section{Conclusions} \label{SEC_CONCL}

\noindent In this work we have analyzed a mathematical model to understand the dynamics of infection of healthy T cells by HIV and SARS-CoV-2 under the effect of HAART. We have shown that the solutions of model \eqref{modelo} are positive and bounded within a biologically reasonable region (Theorem \ref{thm_bound}). Furthermore, we have computed the disease-free and endemic equilibria for HIV and SARS-CoV-2 submodels and their respective basic reproduction numbers, $\mathcal{R}_{\text{hiv}}$ and $\mathcal{R}_{\text{sars}}$. We have proved that the endemic equilibria of each submodel undergo a transcritical bifurcation when the respective basic reproduction number equals one (Theorem \ref{bifurc}). Hence, for submodels \eqref{hiv_submodel} and \eqref{sarscov_submodel}, we have shown that the disease-free equilibria are stable when the respective basic reproduction number is less than 1 and unstable otherwise. Their endemic equilibria are stable when the basic reproduction number is greater than one and unstable otherwise. Finally, regarding system \eqref{modelo}, we have shown that the basic reproduction number for the coinfection model is expressed as $\mathcal{R}_0 = \max{\left\{ \mathcal{R}_{\text{hiv}}, \mathcal{R}_{\text{sars}} \right\} ,}$ and that the disease-free equilibrium point remains stable when $\mathcal{R}_0 < 1$ (Theorem \ref{FULL_THEO}).

\smallskip


\noindent Lastly, regarding the numerical results, we have concluded that antiretroviral therapy has a significant impact in reducing both HIV viral load and HIV-infected cells (Figure \ref{hiv_submodel_figure}). Moreover, although HAART specifically targets HIV and HIV-infected cells, Figures \ref{coinfection_figure}, \ref{coinfection_figure_2} and \ref{epsilons_VS} have indicated that this therapy also reduces SARS-CoV-2 proliferation and SARS-CoV-2-infected cells. Consequently, HAART has proven highly effective in reducing coinfected cells, with greater reductions observed as the efficacy of HAART improves (Figure \ref{coinfection_figure}).

\smallskip

\noindent {\bf Future work.} In future research, we plan to extend our current coinfection model to examine the effects of SARS-CoV-2 vaccination on HIV progression. While this study focused on the influence of HAART on SARS-CoV-2 dynamics, the new approach will assess how COVID-19 vaccines can impact HIV viral load and the number of HIV-infected cells. The possibility of future research using a fractional order derivative model, which could be similar to the work developed by the authors of \cite{MauricioPinto2021, Naik2024}, is not excluded.

%



\end{document}